\documentclass[journal]{IEEEtran}

\usepackage[cmex10]{amsmath}  
\usepackage{amssymb}
\usepackage{graphicx,color}   

\newtheorem{theorem}{Theorem}
\newtheorem{remark}{Remark}
\newtheorem{lemma}{Lemma}

\begin{document}

\title{The Capacity Region of the Restricted Two-Way Relay Channel with Any Deterministic Uplink}
\author{Lawrence Ong and Sarah J.\ Johnson
\thanks{This work is supported by the Australian Research Council under the grant DP1093114.}
}
\maketitle

\begin{abstract}
This paper considers the two-way relay channel (TWRC) where two users communicate via a relay. For the restricted TWRC where the uplink from the users to the relay is any deterministic function and the downlink from the relay to the users is any arbitrary channel, the capacity region is obtained. The TWRC considered is restricted in the sense that each user can only transmit a function of its message.
\end{abstract}

\section{Introduction}

The two-way relay channel (TWRC), where two users exchange data through a relay, was first investigated by Wu et al.~\cite{wuchoukung05}. In this model the users cannot communicate directly and must exchange data by first transmitting to the relay, called the {\em uplink}. The relay then processes the data in some way and broadcasts to both users on the {\em downlink}. Common applications that can be modeled by the TWRC include satellite communications, cellular communications via a base station, and indoor wireless communications via a router.  Using relays to facilitate data exchange is now moving from theory to practice following their introduction in the 802.16j (WiMAX) standard.

The capacity region of the two-way relay channel was found for the case where the uplink and the downlink are both binary symmetric adder channels \cite{namchung08}.
 Since then the capacity region of more general TWRCs has been found for only a few classes of channel models: (i) the uplink and the downlink are both finite-field adder channels~\cite{ongmjohnsonit11}, and (ii) the uplink and the downlink are both linear finite-field deterministic channels~\cite{avestimehrsezgin10}. For the Gaussian TWRC, results within $\frac{1}{2}$ bit of the capacity have been obtained~\cite{namchunglee09}.

In all the classes of TWRCs where the capacity is known, the uplink
channels are linear.  In this paper, we derive the capacity region
of another class of TWRC, where  (i) the uplink is any deterministic
channel\footnote{This includes the linear finite-field deterministic
model as a special case but also includes non-linear channels.},
(ii) the downlink is any arbitrary channel\footnote{This includes
the finite-field adder and Gaussian channels as special cases.}, and
(iii) the users' channel inputs can only depend on their respective
messages, and not on their received channel outputs (this is
commonly known as the {\em restricted channel}). To the best of our
knowledge, this is the first class of TWRCs with non-linear uplinks
where the capacity region is found.

Deterministic channels can model networks comprising fixed-capacity links, and can approximate channels with extremely low noise. Another advantage of the deterministic approach is that one can focus on the interaction between the signals arriving from different nodes rather than the background noise of the system \cite{avestimehrsezgin10}.

We will show that the capacity region of the restricted TWRC with a deterministic uplink and an arbitrary downlink can be achieved using the compress-forward (CF) coding scheme. This scheme was first proposed by Cover and El Gamal for the single-source single-destination single-relay channel~\cite{covergamal79}. Using CF the relay does not decode the received signals but rather compresses them (using Wyner-Ziv coding), bins the compressed signals (using the random-binning technique), and sends the bin index.
CF was extended to the TWRC by Rankov and Wittneben~\cite{rankovwittneben06} and Schnurr et al.~\cite{schnurroechtering07}.

In this work, we obtain the capacity of the restricted deterministic-uplink arbitrary-downlink TWRC in Section~\ref{sec:TWRC_capacity} by showing that the capacity outer bound derived using the cut-set argument~\cite{fongyeung11} coincides with the capacity inner bound derived using CF~\cite{schnurroechtering07}. In the absence of noise on the uplink, we can set the ``quantization noise'' of CF to zero, i.e., having the relay directly map its received signals to its transmitted signals. In this case, binning is also not required. In the light of this observation, we present an alternative proof of the capacity region in Section~\ref{sec:TWRC_simpler} using a simpler coding scheme.

\section{Channel Model}

\begin{figure}[t]
\centering
\resizebox{5.1cm}{!}{
\begin{picture}(0,0)%
\includegraphics{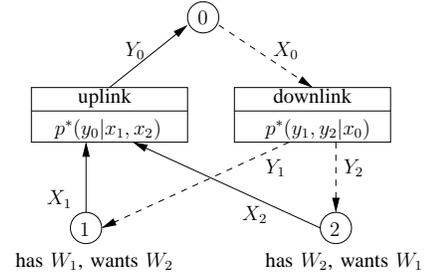}%
\end{picture}%
\setlength{\unitlength}{3947sp}%
\begingroup\makeatletter\ifx\SetFigFont\undefined%
\gdef\SetFigFont#1#2#3#4#5{%
  \fontsize{#1}{#2pt}%
  \fontfamily{#3}\fontseries{#4}\fontshape{#5}%
  \selectfont}%
\fi\endgroup%
\begin{picture}(3627,2635)(1036,-2464)
\put(3526,-361){\makebox(0,0)[lb]{\smash{{\SetFigFont{12}{14.4}{\familydefault}{\mddefault}{\updefault}{\color[rgb]{0,0,0}$X_0$}%
}}}}
\put(1651,-811){\makebox(0,0)[lb]{\smash{{\SetFigFont{12}{14.4}{\familydefault}{\mddefault}{\updefault}{\color[rgb]{0,0,0}uplink}%
}}}}
\put(3526,-811){\makebox(0,0)[lb]{\smash{{\SetFigFont{12}{14.4}{\familydefault}{\mddefault}{\updefault}{\color[rgb]{0,0,0}downlink}%
}}}}
\put(3451,-1486){\makebox(0,0)[lb]{\smash{{\SetFigFont{12}{14.4}{\familydefault}{\mddefault}{\updefault}{\color[rgb]{0,0,0}$Y_1$}%
}}}}
\put(4201,-1486){\makebox(0,0)[lb]{\smash{{\SetFigFont{12}{14.4}{\familydefault}{\mddefault}{\updefault}{\color[rgb]{0,0,0}$Y_2$}%
}}}}
\put(1051,-2386){\makebox(0,0)[lb]{\smash{{\SetFigFont{12}{14.4}{\familydefault}{\mddefault}{\updefault}{\color[rgb]{0,0,0}has $W_1$, wants $W_2$}%
}}}}
\put(3451,-2386){\makebox(0,0)[lb]{\smash{{\SetFigFont{12}{14.4}{\familydefault}{\mddefault}{\updefault}{\color[rgb]{0,0,0}has $W_2$, wants $W_1$}%
}}}}
\put(1426,-1111){\makebox(0,0)[lb]{\smash{{\SetFigFont{12}{14.4}{\familydefault}{\mddefault}{\updefault}{\color[rgb]{0,0,0}$p^*(y_0|x_1,x_2)$}%
}}}}
\put(3451,-1111){\makebox(0,0)[lb]{\smash{{\SetFigFont{12}{14.4}{\familydefault}{\mddefault}{\updefault}{\color[rgb]{0,0,0}$p^*(y_1,y_2|x_0)$}%
}}}}
\put(1666,-2086){\makebox(0,0)[lb]{\smash{{\SetFigFont{12}{14.4}{\familydefault}{\mddefault}{\updefault}{\color[rgb]{0,0,0}$1$}%
}}}}
\put(4066,-2086){\makebox(0,0)[lb]{\smash{{\SetFigFont{12}{14.4}{\familydefault}{\mddefault}{\updefault}{\color[rgb]{0,0,0}$2$}%
}}}}
\put(1351,-1786){\makebox(0,0)[lb]{\smash{{\SetFigFont{12}{14.4}{\familydefault}{\mddefault}{\updefault}{\color[rgb]{0,0,0}$X_1$}%
}}}}
\put(3226,-1936){\makebox(0,0)[lb]{\smash{{\SetFigFont{12}{14.4}{\familydefault}{\mddefault}{\updefault}{\color[rgb]{0,0,0}$X_2$}%
}}}}
\put(2798,-54){\makebox(0,0)[lb]{\smash{{\SetFigFont{12}{14.4}{\familydefault}{\mddefault}{\updefault}{\color[rgb]{0,0,0}$0$}%
}}}}
\put(2101,-361){\makebox(0,0)[lb]{\smash{{\SetFigFont{12}{14.4}{\familydefault}{\mddefault}{\updefault}{\color[rgb]{0,0,0}$Y_0$}%
}}}}
\end{picture}%

}
\caption{The two-way relay channel. For a deterministic uplink, $Y_0 = f (X_1,X_2)$ is a deterministic function.}
\label{fig:twrc} \vspace{-1em}
\end{figure}

The TWRC (see Fig.~\ref{fig:twrc}) consists of three nodes: two users (denoted by nodes 1 and 2) and one relay (denoted by node 0). Let $X_i$ be the channel input from node $i$ and $Y_i$ be the channel output received by node $i$. The general (not necessarily deterministic) TWRC is defined as $p(y_0,y_1,y_2 | x_0,x_1,x_2) = p^*(y_0|x_1,x_2) p^*(y_1,y_2|x_0)$.

We consider $n$ channel uses, and denote the channel variables $X_i$ and $Y_i$ at time $t$ as $X_{i,t}$ and $Y_{i,t}$ respective, for $t \in \{1,2,\dotsc, n\}$. Our block codes consists of the following: (i) An independent message for each user $i$, $W_i \in \{1,2,\dotsc, 2^{nR_i}\}$ for $i \in \{1,2\}$; (ii) Encoding functions for the users, $X_{i,t} = g_{i,t}(W_i)$ for $i \in \{1,2\}$ and $t \in \{1,\dotsc, n\}$; (iii) Encoding functions for the relay, $X_{0,t} = g_{0,t} (Y_{0,1}, Y_{0,2}, \dotsc, Y_{0,t-1})$ for $t \in \{1,\dotsc, n\}$; and (iv) A decoding function for each user, $\hat{W}_j = h_i(W_i,Y_{i,1}, Y_{i,2},\dotsc, Y_{i,n})$. Here $\hat{W}_j$ is the estimate of the message $W_j$ by user $i$, $i \neq j$. Note that in each channel usage, each user transmits a function of its own message (i.e., it is a restricted channel), and the relay transmits a function of its previously received channel outputs.

We say that the rate pair $(R_1,R_2)$ is achievable if the following is true: for any $\epsilon>0$, there exists at least one block code such that $\Pr \{ \hat{W}_1 \neq W_1 \text{ or } \hat{W}_2 \neq W_2\} \leq \epsilon$. The capacity region $\mathcal{C}$ is the closure of all achievable rate pairs.

\section{Capacity Outer Bound and Inner Bound}
We review an outer bound and an inner bound to $\mathcal{C}$. Let
\begin{align}
\mathcal{R}_1 \triangleq \Big\{  &(R_1,R_2) \in \mathbb{R}^2_+:  \nonumber\\
& \quad R_1  \leq I(X_1;Y_0 | X_2) \label{eq:out-1}\\
& \quad R_2 \leq I(X_2;Y_0 | X_1), \label{eq:out-2}\\
& \text{for some } p(x_1,x_2,y_0) =p(x_1)p(x_2) p^*(y_0|x_1,x_2) \Big\}, \nonumber \\
\mathcal{R}_2 \triangleq \Big\{  &(R_1,R_2) \in \mathbb{R}^2_+:  \nonumber\\
& \quad R_1 \leq I(X_0;Y_2) \label{eq:r2-1}\\
& \quad R_2 \leq I(X_0;Y_1), \label{eq:r2-2}\\
&  \text{for some } p(x_0,y_1,y_2) =p(x_0)p^*(y_1,y_2|x_0) \Big\}. \nonumber
\end{align}

Denote the convex hull of a set $\mathcal{R}$ as $\mathsf{Conv}(\mathcal{R})$, and define
\begin{equation}
\mathcal{R}_\text{out} \triangleq \mathsf{Conv}(\mathcal{R}_1) \cap \mathcal{R}_2.
\end{equation}

\begin{remark}
The set $\mathsf{Conv}(\mathcal{R}_1)$ is closed~\cite[pg.\ 625]{shannon61}, and the set $\mathcal{R}_2$ is convex and closed~\cite{fongyeung11}.
\end{remark}

The following outer bound is due to Fong and Yeung~\cite{fongyeung11}:
\begin{lemma} \label{lemma:outer}
$\mathcal{C} \subseteq \mathcal{R}_\text{out}$.
\end{lemma}

Using the CF coding scheme, the following rate region is achievable (i.e., an inner bound to the capacity)~\cite{schnurroechtering07}:
\begin{lemma} \label{lemma:inner}
$\mathcal{R}_\text{CF} \subseteq \mathcal{C}$, where
\begin{align}
\mathcal{R}_\text{CF} \triangleq \Big\{  &(R_1,R_2) \in \mathbb{R}^2_+:  \nonumber\\
& \quad R_1 < I(X_1;\hat{Y}_0|X_2,Q) \label{eq:cf-1}\\
& \quad R_2 < I(X_2;\hat{Y}_0|X_1,Q), \label{eq:cf-2}\\
& \quad \text{subject to the constraints} \nonumber\\
& \quad H(\hat{Y}_0|X_1,Q) - H(\hat{Y}_0|Y_0) < I(X_0;Y_1) \label{eq:cf-3}\\
& \quad H(\hat{Y}_0|X_2,Q) - H(\hat{Y}_0|Y_0) < I(X_0;Y_2),
\label{eq:cf-4}
\end{align}
for some $p(x_0,y_1,y_2) = p(x_0)p^*(y_1,y_2|x_0)$ and \\
$p(q,x_1,x_2,\hat{y}_0,y_0)=p(q)p(x_1|q) p(x_2|q) p(\hat{y}_0 | y_0)
p^*(y_0|x_1,x_2)$ with the cardinality of $Q$ bounded as
$|\mathcal{Q}| \leq 4$  and that for $\hat{Y}_0 \text{ bounded as }
|\hat{\mathcal{Y}}_0| \leq |\mathcal{Y}| + 3\Big\}.$
\end{lemma}

Denote the closure of the set $\mathcal{R}$ by $\overline{\mathcal{R}}$. Since $\mathcal{C}$ is closed,
\begin{equation} \label{eq:cf-capacity}
\mathcal{R}_\text{CF} \subseteq \mathcal{C} \Rightarrow \overline{\mathcal{R}_\text{CF}} \subseteq \overline{\mathcal{C}} = \mathcal{C}.
\end{equation}

The CF achievable region is derived by Schnurr et al.~\cite{schnurroechtering07} for the half-duplex TWRC. The results can be readily extended to the full-duplex TWRC considered in this paper by setting $\alpha=1$ and $\beta=1$. The inequalities in \eqref{eq:cf-1} and \eqref{eq:cf-2} are strict due to the slight difference in the definition of achievable rate pairs in this paper and that in \cite{schnurroechtering07}.

\section{The Capacity of the TWRC with a Deterministic Uplink} \label{sec:TWRC_capacity}

If the uplink of the TWRC is deterministic, we have that
\begin{equation}
p^*(y_0 | x_1,x_2) = \begin{cases} 1, & \text{if } y_0 = f(x_1,x_2)\\
0, & \text{otherwise},
\end{cases} \label{eq:deterministic}
\end{equation}
for some deterministic function $f(x_1,x_2)$. For this channel, we have the following capacity result:
\begin{theorem} \label{theorem}
The capacity region of any restricted deterministic-uplink arbitrary-downlink TWRC is $\mathcal{C} = \mathcal{R}_\text{out}$.
\end{theorem}

\subsection{Points in $\mathsf{Conv}(\mathcal{R}_1)$}

If the uplink is deterministic, we have $H(Y_0 | X_1,X_2) = 0$. So, the RHS of \eqref{eq:out-1} and \eqref{eq:out-2} simplify to $I(X_1;Y_0|X_2) = H(Y_0|X_2) - H(Y_0|X_1,X_2) = H(Y_0|X_2)$ and $I(X_2;Y_0|X_1) = H(Y_0|X_1)$ respectively.
Therefore, we can re-write $\mathcal{R}_1$ as 
\begin{align*}
\mathcal{R}_1 = \Big\{  &(R_1,R_2) \in \mathbb{R}^2_+:  \;  R_1  \leq H(Y_0| X_2),\; R_2 \leq H(Y_0 | X_1), \\ &\;\text{for some } p(x_1,x_2,y_0) =p(x_1)p(x_2)p^*(y_0|x_1,x_2) \Big\}.
\end{align*}

Before proving Theorem~\ref{theorem}, we establish the following:

\begin{lemma} \label{lemma:points}
For some deterministic $p^*(y_0|x_1,x_2)$ as defined in \eqref{eq:deterministic}, any point in $\mathsf{Conv}(\mathcal{R}_1)$ can be written as $(H(Y_0|X_2,Q),H(Y_0|X_1,Q))$ for some $p(q,x_1,x_2,y_0) = p(q) p(x_1|q) p(x_2|q) p^*(y_0|x_1,x_2)$, where $Q$ is an auxiliary random variable with cardinality $|\mathcal{Q}| =3$.
\end{lemma}

\begin{figure}[t]
\centering
\resizebox{4cm}{!}{
\begin{picture}(0,0)%
\includegraphics{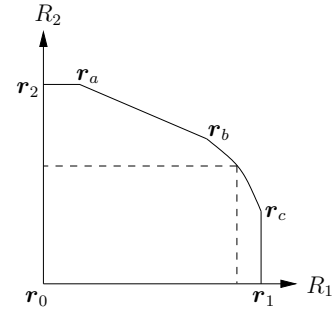}%
\end{picture}%
\setlength{\unitlength}{3947sp}%
\begingroup\makeatletter\ifx\SetFigFont\undefined%
\gdef\SetFigFont#1#2#3#4#5{%
  \reset@font\fontsize{#1}{#2pt}%
  \fontfamily{#3}\fontseries{#4}\fontshape{#5}%
  \selectfont}%
\fi\endgroup%
\begin{picture}(2430,2586)(1261,-2839)
\put(1771,-921){\makebox(0,0)[lb]{\smash{{\SetFigFont{12}{14.4}{\familydefault}{\mddefault}{\updefault}{\color[rgb]{0,0,0}$\boldsymbol{r}_a$}%
}}}}
\put(1351,-2761){\makebox(0,0)[lb]{\smash{{\SetFigFont{12}{14.4}{\familydefault}{\mddefault}{\updefault}{\color[rgb]{0,0,0}$\boldsymbol{r}_0$}%
}}}}
\put(3676,-2686){\makebox(0,0)[lb]{\smash{{\SetFigFont{12}{14.4}{\familydefault}{\mddefault}{\updefault}{\color[rgb]{0,0,0}$R_1$}%
}}}}
\put(1426,-436){\makebox(0,0)[lb]{\smash{{\SetFigFont{12}{14.4}{\familydefault}{\mddefault}{\updefault}{\color[rgb]{0,0,0}$R_2$}%
}}}}
\put(1276,-1036){\makebox(0,0)[lb]{\smash{{\SetFigFont{12}{14.4}{\familydefault}{\mddefault}{\updefault}{\color[rgb]{0,0,0}$\boldsymbol{r}_2$}%
}}}}
\put(3226,-2761){\makebox(0,0)[lb]{\smash{{\SetFigFont{12}{14.4}{\familydefault}{\mddefault}{\updefault}{\color[rgb]{0,0,0}$\boldsymbol{r}_1$}%
}}}}
\put(3331,-2056){\makebox(0,0)[lb]{\smash{{\SetFigFont{12}{14.4}{\familydefault}{\mddefault}{\updefault}{\color[rgb]{0,0,0}$\boldsymbol{r}_c$}%
}}}}
\put(2861,-1391){\makebox(0,0)[lb]{\smash{{\SetFigFont{12}{14.4}{\familydefault}{\mddefault}{\updefault}{\color[rgb]{0,0,0}$\boldsymbol{r}_b$}%
}}}}
\end{picture}%
}
\caption{The region of $\mathsf{Conv}(\mathcal{R}_1)$, which is the convex hull of $\mathcal{R}(X_1,X_2)$ for all $(X_1,X_2) \sim p(x_1)p(x_2)$. The dotted rectangle is an example of $\mathcal{R}(X'_1,X'_2)$ for some $p'(x_1)p'(x_2)$.}
\label{fig:r1} \vspace{-1em}
\end{figure}

\begin{proof}[Proof of Lemma~\ref{lemma:points}]
The region of the form  $\mathsf{Conv}(\mathcal{R}_1)$ is depicted in Fig.~\ref{fig:r1}. Define the {\em top-right boundary} as the boundary segment from $\boldsymbol{r}_2$ to $\boldsymbol{r}_1$ (via $\boldsymbol{r}_a$, $\boldsymbol{r}_b$, and $\boldsymbol{r}_c$), excluding $\boldsymbol{r}_2$ and $\boldsymbol{r}_1$.
For a convex set $\mathcal{R}$, we call any point $\boldsymbol{r} \in \mathcal{R}$ such that $\boldsymbol{r} \notin \mathsf{Conv}(\mathcal{R}) \setminus \{ \boldsymbol{r}\}$ a {\em vertex} of $\mathcal{R}$. In other words, a vertex cannot be formed by taking a weighted average of other points in $\mathcal{R}$. For example, in Fig.~\ref{fig:r1}, $\boldsymbol{r}_0$, $\boldsymbol{r}_1$, $\boldsymbol{r}_2$, $\boldsymbol{r}_a$, $\boldsymbol{r}_b$, $\boldsymbol{r}_c$, and all points from $\boldsymbol{r}_b$ to $\boldsymbol{r}_c$ on the top-right boundary are the vertices. 

Recall that $\mathsf{Conv}(\mathcal{R}_1)$ is the convex hull of the union of rectangular regions of the form $\mathcal{R}(X'_1,X'_2) \triangleq \big\{ 0 \leq R_1 \leq H(Y'_0|X'_2), 0 \leq R_2 \leq H(Y'_0|X'_1) \big\}$ for all $(X'_1,X'_2) \sim p'(x_1)p'(x_2)$. So, any vertex $\boldsymbol{r}$ on the top-right boundary of $\mathsf{Conv}(\mathcal{R}_1)$ must be a vertex of some rectangle $\mathcal{R}(X_1,X_2)$ [i.e., $(H(Y_0|X_2),H(Y_0|X_1))$], because $\boldsymbol{r}$ cannot be written as a weighted average of other points in $\mathsf{Conv}(\bigcup_{X_1,X_2} \mathcal{R}(X_1,X_2)) \setminus \boldsymbol{r}$. So, each of them must be $(H(Y'_0|X'_2),H(Y'_0|X'_1))$ for some $p'(x_1)p'(x_2)$.

We now further show that $\boldsymbol{r}_0$, $\boldsymbol{r}_1$, and $\boldsymbol{r}_2$ can each be written as $(H(Y_0|X_2),H(Y_0|X_1))$ for some $p(x_1)p(x_2)$. First, the largest value for $R_1$ in $\mathsf{Conv}(\mathcal{R}_1)$ cannot exceed
\begin{subequations}
\begin{align}
r_1^\text{max} &= \max_{p(x_1)p(x_2)} H(Y_0|X_2) \\ &= \max_{p(x_1)p(x_2)} \sum_{x_2} p(x_2)  H(Y_0|X_2=x_2) \\ &= \max_{x_2,p(x_1)} H(Y_0|X_2=x_2).
\end{align}
\end{subequations}
Let $X'_1 \sim p'(x_1)$ and $X'_2 = x'_2$ attain $r_1^\text{max}$. Using this choice of input distribution, we have $H(Y'_0 | X'_1) = 0$. So,  $(r_1^\text{max},0)= (H(Y'_0|X'_2),H(Y'_0|X'_1)) \in \mathsf{Conv}(\mathcal{R}_1)$. Since the region $\mathsf{Conv}(\mathcal{R}_1)$ is convex, $\boldsymbol{r}_1$ must attain the largest value of $R_1$ in $\mathsf{Conv}(\mathcal{R}_1)$. Therefore, $\boldsymbol{r}_1 =  (H(Y'_0|X'_2),H(Y'_0|X'_1))$. Similarly, by swapping the role of $X_1$ and $X_2$, we can show that $\boldsymbol{r}_2$ can be written as  $(H(Y''_0|X''_2)=0,H(Y''_0|X''_1))$ for some $X''_1 = x''_1$ and $X''_2 \sim p''(x_2)$. Fixing $X'''_1 = x'''_1$ and $X'''_2=x'''_2$, $\boldsymbol{r}_0 = (H(Y'''_0|X'''_2=x'''_2),H(Y'''_0|X''_1=x'''_1))$.

We have shown that all vertices in $\mathsf{Conv}(\mathcal{R}_1)$ can be written as $(H(Y_0|X_2),H(Y_0|X_1))$ for some $p(x_1)p(x_2)$. Since $\mathsf{Conv}(\mathcal{R}_1)$ is a closed convex set in a two-dimensional space, any point on its boundary can be written as a weighted average of two vertices. From Fig.~\ref{fig:r1}, we see that any interior point in $\mathsf{Conv}(\mathcal{R}_1)$ can be written as a weighted average of a boundary point and $\boldsymbol{r}_0$. Hence, any point $(r_1,r_2) \in \mathsf{Conv}(\mathcal{R}_1)$ can be written as the weighted average of three vertices (the last one being $\boldsymbol{r}_0$), i.e.,
\begin{subequations}
\begin{align}
r_1 & = \sum_{q \in \{a,b,c\}} p(q) H(Y_0|X_2)_{ p^q(x_{1}) p^q(x_{2})} \\
&= \sum_{q} p(q) H(Y_0|X_2, Q=q) \label{eq:q1} \\ &= H(Y_0 | X_2,Q), \label{eq:q2}
\end{align}
\end{subequations}
and similarly,
\begin{equation}
r_2 = H(Y_0|X_1,Q), \label{eq:q3}
\end{equation}
where \eqref{eq:q1}, \eqref{eq:q2}, and \eqref{eq:q3} are evaluated with some $p(q,x_1,x_2,y_0) = p(q) p(x_1|q)p(x_2|q)p^*(y_0|x_1,x_2)$, and $Q$ is an auxiliary (time-sharing) random variable with cardinality $|\mathcal{Q}|=3$.
\end{proof}

\subsection{Proof of Theorem~\ref{theorem}}

With Lemma~\ref{lemma:points}, we now prove Theorem~\ref{theorem}. First, define $\mathcal{R}_2^-$ to be $\mathcal{R}_2$ where \eqref{eq:r2-1} and \eqref{eq:r2-2} are strict inequalities. Since $\mathcal{R}_2$ is closed, $\overline{\mathcal{R}_2^-} = \mathcal{R}_2$. Defining $\mathcal{R}_\text{out}^- \triangleq \mathsf{Conv}(\mathcal{R}_1) \cap \mathcal{R}_2^-$, we have
\begin{equation}
\overline{\mathcal{R}_\text{out}^-} = \overline{ \mathsf{Conv}(\mathcal{R}_1) \cap \mathcal{R}_2^-} = \mathsf{Conv}(\mathcal{R}_1) \cap \overline{\mathcal{R}_2^-}  = \mathcal{R}_\text{out}.
\end{equation}

Next, define $\mathcal{R}_\text{CF}^+$ as $\mathcal{R}_\text{CF}$ where \eqref{eq:cf-1} and \eqref{eq:cf-2} are inequalities (not necessarily strict). For deterministic uplink and choosing $\hat{Y}_0= Y_0$, the RHS of \eqref{eq:cf-1} and \eqref{eq:cf-2} become $I(X_1;\hat{Y}_0|X_2,Q) = H(Y_0|X_2,Q) - H(Y_0|X_1,X_2,Q) = H(Y_0|X_2,Q)$ and $I(X_2;\hat{Y}_0|X_1,Q) = H(Y_0|X_1,Q)$ respectively. So,

{$\;$}
\vspace{-3ex}
\begin{align}
\mathcal{R}_\text{CF}^+ = &\Big\{  (R_1,R_2) \in \mathbb{R}^2_+:  \nonumber\\
& \quad R_1 \leq H(Y_0|X_2,Q), \; R_2  \leq  H(Y_0|X_1,Q), \label{eq:cf-7} \\
& \quad \text{subject to the constraints} \nonumber\\
& \quad H(Y_0|X_1,Q)< I(X_0;Y_1), \label{eq:cf-5} \\
&\quad  H(Y_0|X_2,Q) < I(X_0;Y_2), \label{eq:cf-6} \\
& \quad \text{for some } p(q)p(x_1|q) p(x_2|q) p(\hat{y}_0 | y_0)p^*(y_0|x_1,x_2), \nonumber \\
&\quad p(x_0)p^*(y_1,y_2|x_0), \text{ and } |\mathcal{Q}| \leq 4\Big\}. \nonumber
\end{align}
$\mathcal{R}_\text{CF}^+$ is chosen to include some limit points of $\mathcal{R}_\text{CF}$. Due to constrains \eqref{eq:cf-5} and \eqref{eq:cf-6}, $\mathcal{R}_\text{CF}^+$ might not include all limit points of $\mathcal{R}_\text{CF}$. So, we have that $\overline{\mathcal{R}_\text{CF}^+} = \overline{\mathcal{R}_\text{CF}}$.\footnote{$\overline{\mathcal{R}_\text{CF}^+} = \overline{\mathcal{R}_\text{CF}}$ does not hold in general if we change the strict inequalities in \eqref{eq:cf-3} and \eqref{eq:cf-4} to inequalities, i.e, relaxing the constraints, which might lead to the inclusion of additional regions $\{(R_1,R_2)\}$ specified by \eqref{eq:cf-1} and \eqref{eq:cf-2}.}

Now, pick any point $(r_1,r_2) \in \mathcal{R}_\text{out}^-$. Since $(r_1,r_2) \in \mathsf{Conv}(\mathcal{R}_1)$, using Lemma~\ref{lemma:points}, $r_1$ and $r_2$ can be written as $H(Y_0|X'_2,Q')$ and $H(Y_0|X'_1,Q')$ respectively, evaluated with some $p'(q)p'(x_1|q)p'(x_2|q)$. Also, since $(r_1,r_2) \in \mathcal{R}_2^-$, we have that  $r_1= H(Y_0|X'_2,Q') < I(X'_0;Y_2)$ and $r_2 = H(Y_0|X'_1,Q') < I(X'_0;Y_1)$ for some $p'(x_0)$. Hence $(r_1,r_2) \in \mathcal{R}_\text{CF}^+$. This means $\mathcal{R}_\text{out}^- \subseteq \mathcal{R}_\text{CF}^+$, which implies that $\overline{\mathcal{R}_\text{out}^-} \subseteq \overline{\mathcal{R}_\text{CF}^+}$.

Now, $\mathcal{R}_\text{out} = \overline{\mathcal{R}_\text{out}^-} \subseteq \overline{\mathcal{R}_\text{CF}^+} = \overline{\mathcal{R}_\text{CF}}$. From Lemma~\ref{lemma:outer}, we have $\mathcal{C} \subseteq \mathcal{R}_\text{out}$; from \eqref{eq:cf-capacity}, we have $\overline{\mathcal{R}_\text{CF}} \subseteq \mathcal{C}$. Hence, we have $\mathcal{C} = \mathcal{R}_\text{out}$. This proves Theorem~\ref{theorem}. $\hfill \blacksquare$

\section{A Simpler Proof for Achievability}  \label{sec:TWRC_simpler}

In the previous section, we set $\hat{Y}_0 = Y_0$ for CF. This suggests that we do not actually need to invoke Wyner-Ziv coding and binning in CF. Here, we present a simpler proof (achievability) for the capacity region. Consider $B$ blocks, each containing $n$ channel uses. The $\ell$-th block of the uplink transmission and the $(\ell+1)$-th block of the downlink transmissions are as follows, for all $\ell \in \{1,2,\dotsc, B-1\}$:
\begin{itemize}
\item Choose some $p(x_0)$, $p(x_1)$, and  $p(x_2)$.
\item User $i$ generates $2^{nR_i}$ length-$n$ sequences of $\boldsymbol{x}_i$ with each element randomly and independently selected according to $p(x_i)$, for $i \in \{1,2\}$. It transmits $\boldsymbol{x}_i(w_i)$. 
\item The relay receives $\boldsymbol{y}_0$, which is a deterministic function of $w_1$ and $w_2$. For the chosen codebooks $\{\boldsymbol{x}_1(w_1)\}$ and $\{\boldsymbol{x}_2(w_2)\}$, there are at most $2^{n(R_1+R_2)}$ unique sequences of $\boldsymbol{y}_0$. Index them $\boldsymbol{y}_0(w_0)$, for $w_0 \in \{1,2,\dotsc, M\}$, where $M \leq 2^{n(R_1+R_2)}$.
\item The relay randomly generates $M$ length-$n$ sequences of $\boldsymbol{x}_0$ with each element randomly and independently selected according to $p(x_0)$. It transmits $\boldsymbol{x}_0(w_0)$.
\item User 1 decodes $w_0$ from its received symbols $\boldsymbol{y}_1$ and its own message $w_1$. It can do so reliably\footnote{With arbitrarily small error probability} if $R_2 < I(X_0;Y_1)$. We present the proof in Appendix~\ref{appendix:a}.
\item Using the same argument, user 2 can reliably decode $w_0$ if $R_1 < I(X_0;Y_2)$.
\item From $w_0$, both users know $\boldsymbol{y}_0$ exactly. User 1 can recover $w_2$ reliably and user 2 can recover $w_1$ reliably if $R_1 < I(X_1;Y_0|X_2)$ and $R_2 < I(X_2;Y_0|X_1)$. We present the proof in Appendix~\ref{appendix:b}.
\item Each user $i$ transmits $(B-1)$ messages, each of size $2^{nR_i}$, over $Bn$ channel uses. If the users can reliably decode their requested messages, the rate $(\frac{R_1(B-1)}{B}, \frac{R_1(B-1)}{B})$ is achievable. Choosing a sufficiently large $B$, this scheme achieves rates arbitrarily close to $(R_1,R_2)$.
\end{itemize}
So, this scheme achieves any rate pair $(R_1,R_2) \in \mathcal{R}_4$, where
\begin{align}
\mathcal{R}_4 = \Big\{  &(R_1,R_2) \in \mathbb{R}^2_+:  \nonumber\\
& \quad R_1 < \min \{ I(X_1;Y_0|X_2), I(X_0;Y_2)\}\\
& \quad R_2 < \min \{I(X_2;Y_0|X_1), I(X_0;Y_1)\},\\
& \quad \text{for some } p(x_1,x_2,y_0) = p(x_1) p(x_2) p^*(y_0|x_1,x_2) \nonumber \\
& \quad \text{and } p(x_0,y_1,y_2) =p(x_0)p^*(y_1,y_2|x_0) \Big\}. \nonumber
\end{align}
Using time sharing, the region $\mathsf{Conv}(\mathcal{R}_4)$ is achievable. Without the Wyner-Ziv coding constraints \eqref{eq:cf-3} and \eqref{eq:cf-4}, $\mathcal{R}_4$ is in the same form as $\mathcal{R}_1 \cap\mathcal{R}_2$.  We can show that $\overline{\mathsf{Conv}(\mathcal{R}_4 )} = \mathcal{R}_\text{out} = \mathcal{C}$.

\begin{remark}
One can also show that noisy network coding (NNC)~\cite{limkimelgamal11}, which generalizes CF without using Wyner-Ziv coding, with $\hat{Y}_0=Y_0$ also achieves $\mathsf{Conv}(\mathcal{R}_4 )$. Using NNC, encoding is also performed over $B$ blocks (with $B$ being sufficiently large), but decoding at the users is performed only after the last transmission block. Hence, NNC incurs a larger transmission-to-decoding delay and involves a more complicated decoding scheme (simultaneous decoding over all $B$ blocks) compared to the coding scheme described here.
\end{remark}

\section{The Unrestricted Case}

For the unrestricted TWRC, the problem is hard even with deterministic links. Suppose the uplink is a multiplier channel $Y_0 = X_1 X_2$, where $X_1,X_2 \in \{0,1\}$, and the downlink consists of two noiseless orthogonal channels: $X_0 = (X_0^a, X_0^b)$,  $Y_1 = X_0^a$, $Y_2 = X_0^b$, where $X_0^a, X_0^b \in \{0,1\}$. Without loss of optimality\footnote{Since the downlink is noiseless, any processing that the relay could have done (e.g., decoding or quantizing) can also be done at the users.}, the relay transmits $X_0^a = X_0^b = Y_0$. With this, we convert the unrestricted TWRC to the (unrestricted) multiplying two-way channel~\cite[pg.\ 634]{shannon61} where the capacity remains unknown to date.


\appendices

\section{} \label{appendix:a}
Define the following subset of indices of the relay's codewords $\boldsymbol{x}_0(w_0)$, which have a one-to-one mapping to its received symbols $\boldsymbol{y}_0(w_0)$:
\vspace{-0.8ex}
\begin{align*}
\mathcal{S}_{w_1}(a) = \Big\{ &w_0 \in \{1,2,\dotsc, M\}: \boldsymbol{y}_0(w_0) = f( \boldsymbol{x}_1(a),\boldsymbol{x}_2(k)),\\ &\text{for some } k \in \{1,2,\dotsc,2^{nR_2}\} \Big\}.
\end{align*}
Fixing user 1's message $w_1 = a$, $\mathcal{S}_{w_1}(a)$ is the set of indices $w_0$ for which the corresponding $\boldsymbol{y}_0(w_0)$ are possible uplink outputs when the inputs are $\boldsymbol{x}_1(a)$ and $\boldsymbol{x}_2(k)$ for some $k$. Clearly, $|\mathcal{S}_{w_1}(a)| \leq 2^{nR_2}$ for any $a \in \{1,2,\dotsc, 2^{nR_1}\}$.

Assuming that user 1 sent $\boldsymbol{x}_1(w_1 = a)$, and the relay sent $\boldsymbol{x}_0(w_0=b)$. Receiving $\boldsymbol{y}_1$, user 1 declares that $\hat{w}_0$ was sent by the relay if it can find a unique index $\hat{w}_0 \in \mathcal{S}_{w_1}(a)$ where $(\boldsymbol{x}_0(\hat{w}_0),\boldsymbol{y}_1) \in \mathcal{A}_\epsilon^{(n)}(X_0,Y_1)$, where $\mathcal{A}_\epsilon^{(n)}(X_0,Y_1)$ is the set of jointly typical sequences~\cite[pg.\ 195]{coverthomas06}. Define the following events that can lead to incorrect decoding of $w_0$:\\
\indent\indent $ E_1:(\boldsymbol{x}_0(b),\boldsymbol{y}_1) \notin \mathcal{A}_\epsilon^{(n)}(X_0,Y_1)$\\
\indent\indent $E_2: (\boldsymbol{x}_0(c),\boldsymbol{y}_1) \in \mathcal{A}_\epsilon^{(n)}(X_0,Y_1)$, for some $c \neq b$.

By definition, $b \in \mathcal{S}_{w_1}(a)$, and using the asymptotic equipartition property (AEP)~\cite[pg.\ 196]{coverthomas06}, we have that $\Pr\{E_1\} < \epsilon$, and
\begin{subequations}
\begin{align}
&\Pr\{E_2\} \leq \sum_{c \in \mathcal{S}_{w_1}(a) \setminus \{b\}} (\boldsymbol{x}_0(c),\boldsymbol{y}_1) \in \mathcal{A}_\epsilon^{(n)}(X_0,Y_1) \label{eq:union}\\
& \leq \sum_{c \in \mathcal{S}_{w_1}(a) \setminus \{b\}} 2^{-n(I(X_0;Y_1)-3\epsilon)}  \label{eq:aep}\\
& \leq (2^{nR_2}-1) 2^{-n(I(X_0;Y_1)-3\epsilon)}  < 2^{n(R_2 - I(X_0;Y_1) + 3\epsilon)} \leq 2^{-n\epsilon}, \nonumber
\end{align}
\end{subequations}
if $R_2 \leq I(X_0;Y_1) - 4 \epsilon$. Here, \eqref{eq:union} is due to the union bound, and \eqref{eq:aep} is due to the AEP. Hence, if $R_2 < I(X_0;Y_1)$, we can choose $\epsilon$ and $n$ such that  $\Pr\{\hat{w}_0 \neq w_0\} \leq \Pr\{E_1\} + \Pr \{E_2\} < \epsilon + 2^{-n\epsilon} \triangleq \eta$, where $\eta$ is arbitrarily small.

\section{} \label{appendix:b}

We use the results of the two-way channel~\cite{shannon61}, where
two nodes, say $A$ and $B$, exchange data through the channel
$p^*(y_A,y_B|x_A,x_B)$. It has been shown that the nodes can
reliably exchange data if node $A$ transmits at rate $R_A <
I(X_A;Y_B|X_B)$ and node $B$ transmits at rate $R_B <
I(X_B;Y_A|X_A)$ for some $p(x_A)p(x_B)$ \cite{shannon61}. To apply
this result to the TWRC, we set $R_A = R_1$, $R_B=R_2$, $X_A=X_1$,
$X_B=X_2$, and $Y_0=Y_A=Y_B$. Knowing $Y_0 = Y_A = Y_B$, user 1 can
reliably decode $w_2$ and user 2 can reliably decode $w_1$ if
$R_1 < I(X_1;Y_0|X_2)$ and $R_2 < I(X_2;Y_0|X_1)$.



\end{document}